\documentclass[preprint,12pt,longtitle,authoryear]{elsarticle}
\usepackage{amsmath}
\usepackage{amssymb}
\usepackage{amsthm}
\usepackage{amstext}
\usepackage{amsfonts}
\usepackage{fixltx2e}
\usepackage{textcomp}
\usepackage{verbatim}
\usepackage[english]{babel}
\usepackage{pifont}
\usepackage{color}
\usepackage{lscape}
\usepackage[normalem]{ulem}
\usepackage{booktabs}
\usepackage{float}
\usepackage{latexsym}
\usepackage{epsfig}
\usepackage{graphicx}
\usepackage{bm}
\usepackage{array}
\usepackage{ifthen}
\usepackage[ansinew]{inputenc}
\usepackage{subfigure}

\newcommand {\debeq}	{\begin{eqnarray*}}
\newcommand {\fineq}	{\end{eqnarray*}}

\newcommand	{\PP}{\mathbb{P}}
\newcommand	{\EE}{\mathbb{E}}

\newtheorem{theorem}{Theorem}
\newtheorem{lemma}[theorem]{Lemma}

     {\begin{list}{}%
             {\setlength{\leftmargin}{#1}}%
             \item[]%
     }
     {\end{list}}
		
\title{Consistency of Bayesian inference of resolved phylogenetic trees}
\author{Mike Steel}
\ead{mike.steel@canterbury.ac.nz}

\address{Allan Wilson Centre for Molecular Ecology and Evolution, Department of Mathematics and Statistics, University of Canterbury, Christchurch, New Zealand. Email: {\tt mike.steel@canterbury.ac.nz}, Phone: +6421329705.}

\begin{document}

\begin{abstract}
Bayesian inference is now a leading technique for reconstructing phylogenetic trees from aligned sequence data. In this short note, we formally show that the maximum posterior tree topology provides a statistically consistent estimate of a fully-resolved evolutionary tree under a wide variety of conditions.  This includes the inference of 
gene trees from aligned sequence data across the entire parameter range of branch lengths, and under general conditions on priors in models where the usual `identifiability' conditions hold.  We extend this to the inference of species trees from sequence data, where the gene trees constitute `nuisance parameters', as in the program *BEAST.  
This note also addresses earlier concerns  raised in the literature questioning the extent to which statistical consistency for Bayesian methods might hold in general. 
\end{abstract}

\begin{keyword}
Bayesian phylogenetics, statistical consistency, gene tree, species tree
\end{keyword}

\maketitle
\newpage
\section{Introduction}
Bayesian inference has become a mainstream approach for inferring phylogenetic tree topology from aligned DNA sequence data \citep{lem}.  The approach has a number of desirable features, and there exist powerful software packages for analysing genetic sequence data in this way.  At the same time, some potential theoretical limitations of Bayesian phylogenetics have been identified and studied.  These include potential 
problems with the convergence of MCMC-based Bayesian methods  \citep{mos}, and  properties that appear  to be surprising at first, such as the Bayesian star `paradox' \citep{bay1, bay2, bay3}.

A further property of Bayesian phylogentic inference was raised in a simulation study of  \citet{kol}, suggesting  that Bayesian methods applied to unresolved four-leaf trees (with a zero-length interior edge) with certain combinations of long/short pendant branches
 tended to show increasing bias towards one of three particular resolved trees as the sequence length increased. By  contrast, maximum likelihood was found to favour each of the three resolutions equally.    \citet{kol} initially  suggested the possibility  that for data generated by a resolved four-leaf tree with a certain combination of short and long edges,  Bayesian inference might even be statistically inconsistent (i.e. the tree with the highest posterior probability for the data being different from the tree that generated the data, with a probability that does not tend to zero as the sequence length grows) even for models for which maximum likelihood is known to be statistically consistent \citep{cha}. While \citet{kol} stepped back from this suggestion in a subsequent correction to their original paper, the issue drew attention to  a lack of a  formal proof of the statistical consistent of Bayesian inference for in molecular phylogenetics. 
We provide this here by establishing a more general result that includes the phylogenetic setting as a particular case.

This enhanced generality serves a further purpose,  as it allows us to establish formally the  statistical consistency of Bayesian species tree estimation directly from sequence data 
where the gene trees (and their branch lengths) are treated as a further `nuisance parameters'  (as in the program *BEAST \citep{starbeast}). 

While it might be possible that these results could be derived from other theoretical results in  Bayesian statistics, we provide here a self-contained and essentially elementary proof here, that is tailored towards easy application in the phylogenetic setting. This follows the spirit of Joseph Chang's tailored version of Wald's theorem   that provided a convenient tool to check and establish the consistency of maximum likelihood in phylogenetics \citep{cha}, and which curtailed an unproductive debate in the literature about whether the detailed theoretical assumptions of Wald's original theorem applied.

\section{A general result}

Consider the general problem  of identifying  a discrete parameter lying in an arbitrary finite set $A$ from  a sequence of independent and identically distributed (i.i.d.) observations that take values in an arbitrary finite set $U$. Suppose further that  the probability distribution on $U$ is determined not just by the discrete parameter $a \in A$ but also by some additional (`nuisance') parameters.  In this paper,  we will assume that these additional parameters  are continuous, and if we denote the parameter space associated
with each discrete parameter $a \in A$ -- which we denote by $\Theta(a)$-- is an open subset of some Euclidean space.

In the usual phylogenetic setting, $A$ is the set of fully resolved (binary) phylogenetic tree topologies on a given leaf set, $U$ is the set of possible site patterns,
and the parameter set $\Theta(a)$ specifies, for the tree topology $a$ the  branch lengths of the tree  each of which lies in the range $(0, \infty)$, and possibly other parameters relevant to the model.  Thus, if  we are only concerned with branch lengths then $\Theta(a) = (0, \infty)^{2n-3}$ where
$n$ is the number of leaves of tree $a$.  The trees in $A$ may be either rooted or unrooted, and for reconstruction we estimate the same type of tree (thus in the rooted case, the branch lengths are assumed to be ultrametric).

Returning to the general set-up, let 
$p_{(a, \theta)}$ denote the probability distribution on some finite set $U$ determined by the discrete-continuous parameter pair $(a, \theta)$.  Suppose we have
a discrete (prior) probability distribution $\pi$ on $A$, and, for each $a \in A$,  a continuous (prior) probability distribution on $\Theta(a)$ with a probability density function $f_a(\theta)$.
We will suppose that the following conditions hold for all $a \in A$:

\begin{itemize}
\item[(C1)]
$\pi(a) >0$;
\item[(C2)]
The density $f_a(\theta)$ is  continuous, bounded and nonzero on $\Theta(a)$;
\item[(C3)]
The function $\theta \mapsto p_{(a, \theta)}(u)$ is continuous and nonzero on $\Theta(a)$ for each $u \in U$;
\item[(C4)]
For all $\theta \in \Theta(a)$, and all $b \neq a$, we have: $\inf_{\theta' \in \Theta(b)} d(p_{(a, \theta)}, p_{(b, \theta')})>0$.
\end{itemize}
In (C4) and henceforth, $d$ denotes the $L_1$ metric -- that is, for any two probability distributions $p, q$ on  $U$:
$d(p,q) := \sum_{u \in U} |p(u) - q(u)|.$

 In the phylogenetic setting, if  $\pi$ is any of the usual
non-zero priors on binary phylogenetic trees (e.g. the uniform (`proportional to distinguishable arrangements' or PDA) distribution, or the Yule distribution), then condition (C1) is satisfied.  If we take the usual exponential prior on branch lengths then condition (C2) is satisfied.
For all Markov processes on trees, condition (C3) holds (the nonzero condition holds, since in any tree with pendant edges of positive lengths all site patterns have a strictly positive probability).  Finally, for all models for which identifiability holds (e.g. the general time-reversible (GTR)  model or any submodel down to the highly restrictive Jukes-Cantor  model) condition (C4) holds
(see e.g. \citet{ste}; a specific lower bound on $d$ for the two-state symmetric model  is provided via lemma 7.3 of \citet{ste2}).

Now, suppose we are given a sequence ${\bf u} = (u_1,\ldots, u_k) \in U^k$  generated i.i.d. by some unknown pair $(a, \theta)$ and we wish to identify the discrete parameter ($a$)  from ${\bf u}$ given prior densities on $A$
and the continuous parameters.  The
{\em maximum a-posteriori} (MAP) estimator selects the element $b \in A$ that maximizes the posterior probability of $b$ given ${\bf u}$ -- that is, it maximizes
$\pi(b) \EE_{\theta'}[\PP({\bf u}| b, \theta')]$, where:
\begin{equation}
\label{probeq}
\PP({\bf u}| b, \theta) = \prod_{i=1}^k p_{(b, \theta')}(u_i),
\end{equation}
which is the probability of generating the sequence of i.i.d. observations $(u_1, \ldots, u_k)$ from the underlying parameters $(b, \theta')$, 
and where $\EE_{\theta'}$ refers to taking expectation with respect to the prior probability distribution on $\Theta(b)$.

Let $P(a,\theta, k)$ denote the probability that, for a sequence $u_1,\ldots, u_k$  generated i.i.d. by $(a, \theta)$,
the MAP estimator correctly selects $a$.
The following theorem establishes a sufficient condition for the statistical consistency of the MAP estimator in this context.

\begin{theorem}
\label{mainthm}
Provided conditions (C1)--(C4) hold for all $a \in A$, then
$$\lim_{k \rightarrow \infty} P(a, \theta, k) =1$$  for all $a \in A$, and $\theta \in \Theta(a)$.
\end{theorem}

\begin{proof}
Our proof relies on a general but technical lemma, the proof of which we defer to the Appendix.

\begin{lemma}
\label{lem1}
For any $\epsilon_1, \epsilon_2>0$, and for any $\delta>0$ that is less that a strictly positive value  (determined just by $\epsilon_1$ and $\epsilon_2$)  the following holds:
For  any finite set $U$, and any four
probability distributions $p, q, r, s$ on $U$ that satisfy the three conditions:
\begin{itemize}
\item[(i)]  $d(p,q) \geq \epsilon_1$;
\item[(ii)]  for all $u \in U$ with $r(u)>0$, $p(u) \geq \epsilon_2$ and $q(u)>0$;
\item[(iii)] $d(p,r) < \delta$ and  $d(p,s) < \delta$;
\end{itemize}
the quantity $Q= \sum_{u \in U: r(u)>0} r(u) \log\left(\frac{s(u)}{q(u)}\right)$
is well defined (i.e. logarithms are applied to positive quantities) and $Q \geq \frac{1}{3}\epsilon_1^2$.
\end{lemma}

\subsection{Application to the proof of Theorem~\ref{mainthm}}

To apply Lemma~\ref{lem1} we need to define the quantities mentioned by it, and we will do this in the order  $p, s$  then $q, r$ followed by $\epsilon_1$ and $\epsilon_2$. Notice that the definition of $q, r$ and $s$ depends on the data, so these probability distributions are random variables (they depend on the data), but this causes no problem for the argument as we remark at the end of the proof.

We  suppose throughout that the sequence ${\bf u} = u_1,\ldots, u_k$ is generated i.i.d. by $(a, \theta_0)$ where $\theta_0$ is any particular element of $\Theta(a$). Then
 the  MAP estimator will correctly select $a$ from ${\bf u}$  if and only if  the {\em Bayes Factor} defined by:
$$BF_{a/b} = \frac{\pi(a) \EE_\theta[\PP({\bf u}| a, \theta)]}{\pi(b) \EE_{\theta'}[\PP({\bf u}| b, \theta')]}$$
is strictly greater than 1 for all $b \neq a$.  By the Bonferroni inequality, it suffices to show that for each $b \neq a$ the probability that ${\bf u}$ is such that $BF_{a/b} >1$ tends to $1$ as $k$ grows.
To achieve this we first observe that 
$BF_{a/b} =  \frac{\pi(a)}{\pi(b)} \cdot R_{a/b}$  where:
\begin{equation}
\label{R_{a/b}1}
R_{a/b}:= \frac{ \EE_\theta[\PP({\bf u}| a, \theta)]}{\EE_{\theta'}[\PP({\bf u}| b, \theta')]},
\end{equation}
and where $\frac{\pi(a)}{\pi(b)}$, is finite and strictly positive by (C1).  Thus, it suffices to show that, for each $b \neq a$ and for 
 any finite constant $M$, the inequality
$
R_{a/b}>M 
$
holds with a probability that tends to $1$ as $k \rightarrow \infty$.  We will establish this inequality by providing an explicit lower bound to the numerator of $R_{a/b}$ and an explicit upper bound to the denominator of $R_{a/b}$, and showing that, with probability tending to $1$ as $k$ grows, their ratio exceeds $M$.  

Before describing the lower bound, observe that we can re-write Eqn.~(\ref{probeq}) as follows:
\begin{equation}
\label{probeq2}
\PP({\bf u}| b, \theta) = \prod_{u \in U}p_{(b, \theta)}(u)^{n_u},
\end{equation}
where, for each $u \in U$, $n_u := |\{i: u_i = u\}|$.

For the lower bound on the numerator of $R_{a/b}$,  consider the subset $N_\tau$ of  $\Theta(a)$ consisting of a closed ball
centered on $\theta_0$ and of radius $\tau>0$.  Note that we can always select a sufficiently small value of  $\tau>0$ for which $N_\tau \subset \Theta(a)$ by the assumption that $\Theta(a)$ is an open subset of some Euclidean space.   Letting $\mu(N_\tau) = \int_{N_\tau} f_a(\theta) d\theta >0$ we have:
\begin{equation}
\label{lower1}
\EE_\theta[\PP({\bf u}| a, \theta)]  = \int_{\Theta(a)} \PP({\bf u}| a, \theta)f_a(\theta)d\theta \geq \int_{N_\tau}\PP({\bf u}| a, \theta)f_a(\theta)d\theta \geq \mu (N_\tau) \cdot \inf_{\theta \in N_\tau} \{\PP({\bf u}| a, \theta)\}.
\end{equation}

\subsection{Lower bound and the distributions $p$ and $s$}
Let $p= p_{(a, \theta_0)}$ (the generating probability distribution on the true parameters) and 
let  $s$ be the probability distribution of the form $p_{(a, \theta)}$ that minimizes $\PP({\bf u}| a, \theta)$ when
$\theta$ is restricted to $N_\tau$;  such a distribution $s$ exists from the compactness of $N_\tau$ and the continuity condition of (C3).  Then, from (\ref{probeq2})  we have:
$\inf_{\theta \in N_\tau} \{\PP({\bf u}| a, \theta)\}=    \prod_{u \in U} s(u)^{n_u}.$
Applying this to (\ref{lower1}) gives:
\begin{equation}
\label{lower2}
\EE_\theta[\PP({\bf u}| a, \theta)]  \geq \mu (N_\tau) \cdot  \prod_{u \in U} s(u)^{n_u}.
\end{equation}

\subsection{ Upper bound and the distributions $q$ and $r$} 

Regarding the upper bound on the denominator of $R_{a/b}$, we have:
\begin{equation}
\label{upper1}
\EE_{\theta'}[\PP({\bf u}| b, \theta')]  \leq \sup_{\theta' \in \Theta(b)} \{\PP({\bf u}|b, \theta')\}.
\end{equation}
Given ${\bf u}$, let $\theta_i$ be a sequence of elements of $\Theta(b)$ for which $\lim_{i \rightarrow \infty} \PP({\bf u}|b, \theta_i)= \sup_{\theta' \in \Theta(b)} \{\PP({\bf u}|b, \theta')\}$.

Notice that   $p_{(b, \theta_i)}$, $i \geq 1$, is a sequence in a bounded subset of Euclidean space (the probability simplex) and so, by the Bolzano--Weierstrass theorem, it has a convergent subsequence, with limit
$q$ (a probability distribution on $U$). 

It remains to specify the fourth distribution $r$, which is determined purely by the data,  and records the proportion of occurrences of the various outcomes. That is, for each 
$u \in U$ let $r(u) := \frac{1}{k}n_u$.
Notice that $r = (r(u): u \in U)$ is a (empirical) probability distribution on $U$ (i.e. its entries are nonzero and sum to $1$). In the phylogenetic setting
$r$ describes the frequency of site patterns in the data.

\subsection{Combining the two bounds}

Eqns. (\ref{R_{a/b}1}), (\ref{lower2}) and (\ref{upper1}) gives:
\begin{equation}
\label{Req}
R_{a/b} \geq \frac{ \mu (N_\tau) \cdot  \prod_{u \in U} s(u)^{n_u}}{ \prod_{u \in U} q(u)^{n_u}}.
\end{equation}
By (C3), $p(u)>0$ for all $u \in U$, and by the continuity condition (C4), we can select $\tau>0$ sufficiently small so that $s(u) >0$ for all $u \in U$. 
Suppose there exists some $u_0 \in U$ with $r(u_0)>0$ (i.e. $n_{u_0} \geq 1$) and with $q(u_0)=0$.   Then  Eqn. (\ref{Req}) implies that $R_{a/b} = +\infty$ and so Bayesian inference will select
$a$ over $b$. Otherwise we may assume that $q(u)>0$ for all $u \in U$ for which $r(u)>0$, in which case we can take logarithms of both sides of Eqn. (\ref{Req}) and so 
obtain the fundamental inequality:
\begin{equation}
\label{justit}
\log(R_{a/b}) \geq \log(\mu(N_\tau)) +k \sum_{u \in U: r(u)>0} r(u) \log\left(\frac{s(u)}{q(u)}\right).
\end{equation}

\subsection{Definitions of $\epsilon_1$ and $\epsilon_2$}

 Fix $b\in A-\{a\}$ and let
$\epsilon_1 := \inf_{\theta' \in \Theta(b)} \{d(p, p_{(b, \theta')})\} \mbox{ and let  }   \epsilon_2: = \min\{p(u): u \in U, r(u)>0\}.$
Notice that $\epsilon_1>0$ by (C4) and $\epsilon_2>0$ by (C3).

\subsection{Completing the argument}
Returning to the proof of Theorem~\ref{mainthm},  we are now in a position to apply Lemma~\ref{lem1}.   First observe that parts (i) and (ii) of Lemma~\ref{lem1} hold by definition of
$\epsilon_1$ and $\epsilon_2$, respectively.

Next, observe that the event that $d(p,r) \leq \delta$  has 
probability converging to $1$ as $k$ grows, by the law of large numbers.  Thus, with probability converging to 1 as $k \rightarrow \infty$ the first half of Part (iii) of Lemma~\ref{lem1} holds (i.e. $d(p, r)< \delta$).
Moreover,  by  the continuity condition in (C3), we can select $\tau>0$ sufficiently small so that $d(p,s) < \delta$, and so the second half of Part (iii) of  Lemma~\ref{lem1} also holds.

In summary, with probability converging to 1 as $k$ grows, the conditions of Lemma~\ref{lem1} are satisfied, in which case
(by (\ref{justit}))  $$\log(R_{a/b}) \geq \log(\mu(N_\tau)) + k\cdot \frac{1}{3}\epsilon_1^2.$$
Thus, with probability converging to 1 as $k$ grows, for any finite value $M$,   
$R_{a/b} >M$.    By the comments
following Eqn. (\ref{R_{a/b}1}), this completes the proof. 
\hfill$\Box$.

{\bf Remark} In the proof, notice that only the probability distribution $p = p_{(a,\theta_0)}$ is fixed, the other three distributions $r, s, q$ depend on the data ${\bf u}$ that is generated by $p$. However, Lemma~\ref{lem1}  quantifies over
all choices of $r,s,q$ once the $\epsilon_1$ and $\epsilon_2$ values have been specified, and these two $\epsilon_i$ values  depend ultimately just on $p$ by definition).

\end{proof}

\section{Inferring species trees directly from sequences with gene trees treated as `nuisance parameters'}
Consider a fully resolved species tree with branch parameters corresponding to inter-speciation times, and ancestral population sizes. Such a model induces a probability distribution on gene trees under a process of incomplete lineage sorting that is modelled by the multi-species coalescent model \citep{deg}.  Suppose we generate $N$ independent gene trees under this process, and on each gene tree, we evolve sequence sites under a time-reversible site substitution model in which the branch lengths on the gene tree are (in expectation) are an i.i.d. scalar multiple of the branch lengths in the species tree (i.e. we allow different genes to evolve at different rates, but assume that these rates are chosen independently from a given distribution).  

Now, for any fully resolved species tree $T$ and  any tree $T'$ on the same leaf set  that has a different topology from $T$, there exists at least one  triplet of taxa $x,y,z$, say, for which $T|\{x,y,z\} = xy|z$ and $T'|\{x,y,z\} \neq xy|z$.  Under the multispecies coalescent, if $T$ is the generating species tree, then the probability that the induced gene tree has the topology $xy|z$ is strictly greater than the probability it has one of the other two topologies (which have equal probability) \citep{deg}.   Moreover, for any time-reversible site substitution process, the probability that two taxa are both in a given state (say $0$) is a continuous and strictly monotone decreasing function of the temporal separation between them \citep{ald}.   

Consequently, if $T$ is the generating species tree then the probability that any given sequence site has the same given state ($0$) for taxa $x$ and $y$ is strictly larger than that for $T'$; moreover there is a strictly positive lower bound on this positive difference that applies for any such $T' \neq T$ regardless of its branch lengths, so (C4) holds.  
If we now take $A$ to be the finite set of species tree topologies and $U$ to be site patterns, then the conditions for Theorem~\ref{mainthm} apply and so the posterior probability of the generating species tree converges to 1 as $k \rightarrow \infty$. Notice that $N$ does not need to converge to infinity here, nor does the length of sequences for any one gene; only  the total sequence length ($k$) needs to do so. 

\section{Concluding comments}
 In certain Bayesian implementations, the output tree is not the tree that is most frequently found;  rather, a score is assigned to each cluster (subset of taxa) according to its frequency as a clade in the posterior distribution of trees, and 
a consensus tree is constructed on the clusters with the highest posterior support  \citep{starbeast}.  There are various options here as to how this can be implemented, but it is clear that, in general, such a tree could differ from the MAP tree on a given set of data.  This 
raises an obvious question:  Is this consensus tree constructed from the clusters with highest posterior support as clades a consistent estimator of the true species tree?
In the limit, any such tree will converge on the true tree (and the MAP tree) as $k$ (the sequence length) grows, for the following reason: Since we are assuming that the species tree $T$ is fully resolved and that the posterior probability of $T$ converges to 1 (with increasing $k$), the only clusters that will have a posterior probability greater than any positive value  $\epsilon>0$ for all $k$  will be clades in $T$ (and each clade in $T$ will have posterior support approaching 1 as $k$ grows);  otherwise, if some cluster $C$ not in $T$ had this property, then  another tree $T'$ would exist for which the posterior probability of $T'$ would be at least $\epsilon'>0$ for all $k$, contradicting the assumption that the posterior probability of $T$ converges to 1 as $k \rightarrow \infty$ (we can take $\epsilon'$ to be $\epsilon$ divided by the number of fully resolved trees that contain the cluster $C$ as a clade). 

 For future work, the consistency of phylogenetic questions on nonresolved trees could be of interest, as in that case condition (C4) does not hold. 

\subsection{Acknowledgments} I thank Tandy Warnow and Elchanan Mossel for helpful discussion and encouragement,  and Joe Thornton for clarifying some comments in 
\citet{kol}.

\section{References}

\bibliographystyle{model2-names}
\bibliography{bayes_mike_steel}


\section*{Appendix:   Proof of Lemma~\ref{lem1}}
\label{proof}
\begin{proof}
We will require at the outset that $\delta< \min\{\epsilon_1, \frac{1}{2}\epsilon_2\}$; later we place a third  upper bound on $\delta$.
Applying the triangle inequality to  conditions (ii) and (iii), with   $\delta < \frac{1}{2} \epsilon_2$ implies that $r(u)$ and $s(u)$ are both at least $\frac{1}{2} \epsilon_2$ for all $u \in U$, 
and so $Q$ is well-defined (i.e. logarithms are only applied to positive entries). 
Let $\eta = \min\{q(u): r(u)>0\}$.  By condition (ii), $\eta>0$.  For each $u \in U$, let $\Delta_u:= r(u) - s(u)$.  Then:
\begin{equation}
\label{KL}
\sum_{u \in U} r(u) \log\left(\frac{s(u)}{q(u)}\right)  =\sum_{u \in U} s(u)\log\left(\frac{s(u)}{q(u)}\right)  + \sum_{u \in U} \Delta_u \log\left(\frac{s(u)}{q(u)}\right).
\end{equation}
Now, the first term on the right hand-side of (\ref{KL}) is simply the Kullback-Leibler separation of $s$ and $q$ and, by Pinsker's Inequality \citep{cov},  this is bounded below by
$\frac{1}{2}d(s,q)^2$.  Moreover, by the triangle inequality, $d(s,q) \geq d(p,q) - d(p,s) \geq \epsilon_1-\delta$ (by conditions (i) and (iii)) and since $\delta < \epsilon_1$ (so $\epsilon_1-\delta>0$)  the first term on the right of (\ref{KL}) is
bounded below by $\frac{1}{2}(\epsilon_1-\delta)^2.$  

Concerning the second term on the right of (\ref{KL}), its absolute value is bounded above by:
$$\sum_{u \in U} |\Delta_u| \cdot \max_{u \in U}|\log\left(\frac{s(u)}{q(u)}\right)|=  d(r,s) \cdot \max_{u \in U}|\log\left(\frac{s(u)}{q(u)}\right)|.$$
Again invoking the triangle inequality,  $d(r,s) \leq d(r,p)+d(p,s) \leq 2\delta$ (by condition (iii)).  Moreover, since $s(u) \geq p(u) -\delta\geq \epsilon_2-\delta$ (by condition (ii)) and since $\delta< \epsilon_2$ (so $\epsilon_2-\delta>0$)  and $q(u) \geq \eta$:
$$ \max_{u \in U}|\log\left(\frac{s(u)}{q(u)}\right)| \leq \max_{u \in U} |\log(s(u))| + \max_{u \in U} |\log(q(u))|  \leq |\log(\epsilon_2-\delta)|+|\log \eta |.$$

Thus we select $\delta>0$ sufficiently small (in addition to the earlier two upper bounds on $\delta$) so that
$$\frac{1}{2}(\epsilon_1-\delta)^2 - 2\delta \cdot (|\log(\epsilon_2- \delta)| + |\log \eta|) \geq \frac{1}{3} \epsilon_1^2.$$
then the bounds placed above on  the terms in (\ref{KL})
ensure that $Q \geq \frac{1}{3}\epsilon_1^2$ as required. 

\end{proof}

\end{document}